\documentclass[aap]{imsart}
\usepackage{amsmath,amsfonts,amssymb,amsthm,bbm,enumitem}%,hyperref
\usepackage{graphicx}
\usepackage{floatrow}
\usepackage[para]{footmisc}
\usepackage{hyperref}

\arxiv{arXiv:1707.02181}

\startlocaldefs

\theoremstyle{plain}
\newtheorem{thm}{Theorem}
\newtheorem{prop}{Proposition}[section]
\newtheorem{lemma}[prop]{Lemma}
\newtheorem{cor}[prop]{Corollary}

\theoremstyle{definition}
\newtheorem{rmk}[prop]{Remark}
\numberwithin{equation}{section}

\endlocaldefs

\begin{document}

\begin{frontmatter}

% "Title of the paper"
\title{Real eigenvalues in the non-Hermitian Anderson model}
\runtitle{Non-Hermitian Anderson model}

% indicate corresponding author with \corref{}
\author{\fnms{Ilya} \snm{Goldsheid}\ead[label=e1]{i.goldsheid@qmul.ac.uk}\thanksref{m1}}
\address{School of Mathematical Sciences, \\
Queen Mary University of London, \\ London E1 4NS, \\
United Kingdom. \\ \printead{e1}}
%\affiliation{Queen Mary, University of London}

\and

\author{\fnms{Sasha} \snm{Sodin}\ead[label=e2]{a.sodin@qmul.ac.uk}\thanksref{t2,m1,m2}}
\address{School of Mathematical Sciences, \\
Queen Mary University of London, \\ London E1 4NS, \\
United Kingdom \\ \& \\
School of Mathematical Sciences, \\
Tel Aviv
University, \\
Tel Aviv, 69978, \\ Israel. \\ \printead{e2}}
\thankstext{t2}{Supported in part by the European Research Council starting grant 639305 (SPECTRUM) and by a Royal Society Wolfson Research Merit Award.} 
\affiliation{Queen Mary, University of London\thanksmark{m1} \& Tel Aviv University\thanksmark{m2}}

\runauthor{I.\ Goldsheid \& S.\ Sodin}

\begin{abstract}
The eigenvalues of the Hatano--Nelson non-Hermitian Anderson 
matrices, in the spectral regions in which the Lyapunov exponent exceeds the non-Hermiticity parameter, are shown to be real and exponentially close to the Hermitian eigenvalues.  This complements previous results, according to which the eigenvalues in the spectral regions in which the non-Hermiticity parameter exceeds the Lyapunov exponent are aligned on curves in the complex plane.
\end{abstract}



\end{frontmatter}

\section{Introduction and the main result}

Let $v_1, v_2, \cdots$ be independent, identically distributed random variables (potential) and let $g$ be a real parameter, $g\ge0$.
Consider the $N \times N$ random matrix
\begin{equation}\label{eq:defop}
H_N(g) = \left( \begin{array}{ccccccc}
v_1 &e^{-g}& 0&0&\cdots&0&e^{g} \\
e^{g} & v_2& e^{-g} &0&\cdots&0&0\\
0&e^{g}&v_3&e^{-g}&\cdots&0&0\\
&&&\ddots&&&\\
&&&&\ddots&&\\
0&0&0&0&&v_{N-1}&e^{-g}\\
e^{-g}&0&0&0&\cdots&e^{g}&v_N
\end{array}\right)\end{equation}

\medskip
Non-Hermitian matrices of the form (\ref{eq:defop}) were introduced and
studied by Hatano and Nelson \cite{HN1,HN2} to describe the reaction of an Anderson-localised quantum
particle on a ring to a constant imaginary vector field. For $g = 0$, the matrix $H_N = H_N(0)$ is Hermitian, and the eigenvalues are real. For $g>0$,
the eigenvalues are not necessarily real.  The numerical studies of Hatano and Nelson (carried out for the case when the $v_j$ have the 
uniform $[-1,1]$ distribution) suggest that there exist critical values $\overline{g}_{\textrm{cr}}>\underline{g}_{\textrm{cr}}> 0$ such that the following holds:
\begin{enumerate}[label=(\alph*)]
\item For $0\le g<\underline{g}_\textrm{cr}$, all the eigenvalues of $H_N(g)$ are real;
\item for $g\in (\underline{g}_{\textrm{cr}},\overline{g}_{\textrm{cr}})$, some of the eigenvalues remain
real, while others align along a smooth curve in the complex plane;
\item essentially all eigenvalues move out of the real axis when $g> \overline{g}_{\textrm{cr}}$.
\end{enumerate}

A variant of this numerical experiment in the regime (b) is depicted on Figure~\ref{fig1}.  These observations, and especially (b) and (c), were supported by the subsequent analysis performed on the physical level of rigour;
see especially \cite{BZ,BSB,FZ,Z}.
We refer to these works and also to \cite{Molinari} and references therein
for a discussion  of the properties of the (left and right) eigenvectors of
$H_N(g)$, and for extensions to the strip and to higher dimension, which will mostly
remain outside the scope of this paper (see however Section~\ref{s:comments}).

\medskip
In the mathematical works \cite{GKh,GKh1,GKh2} of Khoruzhenko and the first author it was shown that the behaviour of the eigenvalues depends crucially on the Lyapunov 
exponent $\gamma(E)$ associated to the Hermitian operator (see Section~\ref{s:prel}, equation (\ref{eq:gamma1})).  Let us label the algebraic
spectrum of $H_N(g)$ $\{ \lambda_1(g), \cdots, \lambda_N(g)\}$ so that each $\lambda_j(g)$ is
a continuous function of $g$, and $\lambda_1(0) \geq \cdots \geq \lambda_N(0)$ (cf.\
Lemma~\ref{l:eq} below).   

Fix $j$; for $g = 0$ the eigenvalue $\lambda_j(0)$ lies on the real axis. It was shown in  \cite{GKh,GKh1} that for $g < \gamma(\lambda_j(0))$ the eigenvalue $\lambda_j(g)$ remains in the vicinity of the real axis (i.e.\ it lies in the strip $|\Im \lambda|<\epsilon$, provided that  $N \geq N_0(\epsilon)$), whereas for $g > \gamma(\lambda_j(0))$ it escapes to certain polynomial curves $\Gamma_g^{(N)}$ in the complex plane. These statements hold simultaneously for all the eigenvalues $\lambda_j(g)$ on an event of asymptotically full probability. As $N \to \infty$, $\Gamma_g^{(N)}$ converges to the   curve
$\Gamma_g = \{ z \in \mathbb C \, \mid \, \gamma(z) = g\}$.

\begin{figure}
\floatbox[{\capbeside\thisfloatsetup{capbesideposition={right,center},capbesidewidth=5cm}}]{figure}[\FBwidth]
{\caption{A realisation of the spectrum for a variant of (\ref{eq:defop}) with a $2$-periodic background (cf.\ Remark~\ref{rmk:backgr});  $N = 70$ and $g = .08 \in (\underline{g}_{\textrm{cr}},\overline{g}_{\textrm{cr}})$. The  axis is split in 5 intervals. On the odd ones, the Lyapunov exponent is $\geq .08$; on the even ones, it is $\leq .08$.}\label{fig1}}
{\includegraphics[width=7.2cm]{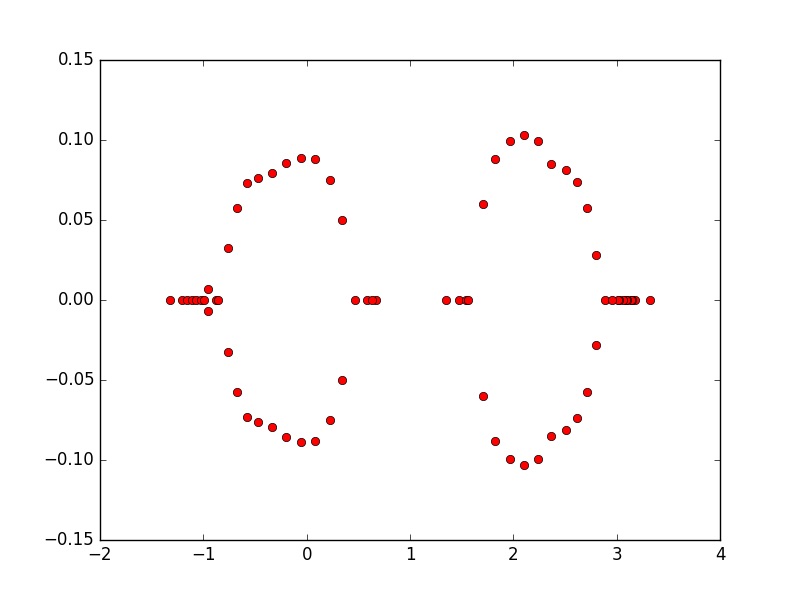}}
\end{figure}

In \cite{GKh2}, these results were extended to a wide class of  deterministic potentials, under the mild assumption of  existence of the integrated density of states $\mathcal{N}(E)$.
Under this assumption, one
defines the Lyapunov exponent via the Thouless formula
\[ \gamma(E) = \int \log|E-E'| \, d\mathcal{N}(E')~.\]
In the case of stationary random sequences, this definition coincides with the
usual one, given in (\ref{eq:gamma1}).
 
 Moreover, it was shown in \cite{GKh2} that the eigenvalues near the curves $\Gamma_g$ boast  regular behaviour on a local scale:
after re-scaling the eigenvalues near a fixed $z \in \Gamma_g$
by the mean (complex) spacing, these align, in the large $N$ limit, on an arithmetic progression.

\medskip
Consequently, the critical values should be given by the formul{\ae}
\[ \underline{g}_\textrm{cr} = \min \{ \gamma(E) \, \mid \, E \in \mathcal S\}~, \quad
\overline{g}_{\textrm{cr}} = \max \{ \gamma(E) \, \mid \, E \in \mathcal S\}~,\]
where $\mathcal S$ is the support of the limiting eigenvalue distribution of $H_N(0)$ (i.e.\ the support of the integrated density of states $\mathcal{N}(E)$ defined in \eqref{eq:ids}, or equivalently the essential spectrum of the infinite-volume self-adjoint operator). 

The results proved in \cite{GKh,GKh1,GKh2} provide a detailed statistical description of the behaviour of the eigenvalue $\lambda_j(g)$ for $g > \gamma(\lambda_j(0))$, both in the global and the local limiting regime; thus one has a complete description of the regime (c), and a partial one -- of (b).

The description of the behaviour for $g < \gamma(\lambda_j(0))$ remained incomplete. In fact, neither the rigorous analysis of \cite{GKh,GKh1,GKh2} nor the heuristic arguments of \cite{BZ,BSB,FZ,Z} provide an indication on whether these eigenvalues are truly real (as suggested by  computer simulations  such as Figure~\ref{fig1}), or they may have a non-zero but asymptotically vanishing imaginary part. 

To the best of our knowledge, no progress on this question has been made since the work \cite{GKh2} had been published. We are also not aware of any previous analysis of the
spacings between these eigenvalues (the local regime).

\medskip In this work we provide a reasonably complete description of the regime $\gamma(\lambda_j(0))>g$, thus settling these two questions. We  prove that in
the case of (\ref{eq:defop}) with independent, identically distributed potential
the corresponding non-Hermitian eigenvalues $\lambda_j(g)$ do in fact
remain on the real axis, and, moreover, they are exponentially close to the Hermitian eigenvalues $\lambda_j(0)$. In other words, if $j$ is fixed and $g$ varies from $0$ to $\infty$, the eigenvalue $\lambda_j(g)$ remains real and exponentially close to $\lambda_j(0)$
for $g \leq \gamma(\lambda_j(0)) - \epsilon$ (where $\epsilon > 0$ is arbitrary small, and $N \geq N_0(\epsilon)$). This complements the result of \cite{GKh2}, according to which $\lambda_j(g)$ aligns near $\Gamma_g$ for $g \geq \gamma(\lambda_j(0)) + \epsilon$.  See Figure~\ref{fig2} for an illustration. 

\begin{figure}
\floatbox[{\capbeside\thisfloatsetup{capbesideposition={right,center},capbesidewidth=3cm}}]{figure}[\FBwidth]
{\caption{The curves connecting the points $(\lambda_j(g), g) \in \mathbb{R}^2$ with $\lambda_j(g) \in \mathbb{R}$, for $N = 70$ and $v_j \sim \operatorname{Unif}[0, 4]$. In the large $N$ limit, the upper envelope of these curves converges to the graph of $\gamma(E)$ on $\mathcal S$. Note that the curves are almost vertical.}\label{fig2}}
{\includegraphics[width=7.5cm]{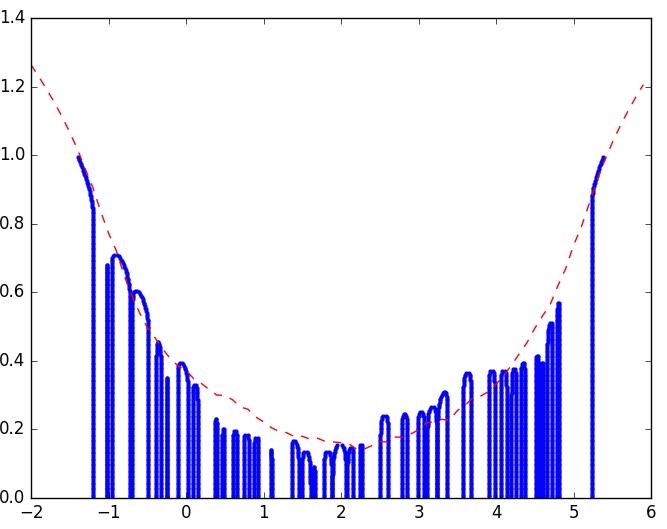}}
\end{figure}

In contrast to the potential-theoretic approach of \cite{GKh2}, our arguments are based on the properties
of products of random matrices.

\begin{thm}\label{thm:realev} Assume that $(v_j)$ is a sequence of i.i.d. random variables and that
$|v_1| \leq A < \infty$ almost surely. Then for any $\epsilon > 0$ 
\[
\mathbb{P} \Big\{ \text{for all $1 \leq j \leq N$ and $g \in [0,  \gamma(\lambda_j(0)) - \epsilon]$ one has:}  \,\,\,
\lambda_j(g) \in \mathbb R \Big\} \underset{N \to \infty}\longrightarrow 1\]
and moreover there exists $c = c(\epsilon) > 0$ such that
\[ \mathbb{P} \Big\{ \forall j \,\,\, \forall g \in [0,  \gamma(\lambda_j(0)) - \epsilon]:  \,\,\,
\lambda_j(g) \in  ( \lambda_j(0) - e^{-cN}, \lambda_j(0)+e^{-cN})  \Big\} \underset{N \to \infty}\longrightarrow 1~.\] 
\end{thm}

\begin{rmk} The first part of the theorem is essentially equivalent to the following statement: if $I$ is an interval, then for any $\epsilon > 0$
\[ \lim_{N \to \infty} \mathbb P \left\{ \text{for all $\lambda_j(0) \in I$ and all $g \leq \inf_{E \in I} \gamma(E) - \epsilon$ one has $\lambda_j(g) \in I$} \right\} = 1~.\] 
\end{rmk}

\begin{rmk} Without invoking new ideas, the theorem can be shown to hold under the
weaker assumption $\mathbb{E} |v_1|^\eta < \infty$ for some $\eta > 0$. We restrict ourselves to the case of bounded random variables, to keep the argument reasonably short. On the other hand, we do insist on avoiding any regularity assumptions on the potential.\end{rmk}

\begin{rmk}\label{rmk:backgr} Only minor adjustments in the argument are required to consider a variant of the model in which $v_j$ is replaced with $v_j + a_j$, where $(a_j)$ is a non-random periodic sequence. For cosmetic reasons, we chose to depict this variant in Figure~\ref{fig1}, which we included for illustration only.
\end{rmk}

\begin{rmk}\label{rmk:ext} Similarly to \cite{GKh1} and in contrast to \cite{GKh2} we assume that $(v_j)$ is an i.i.d.\ sequence. While we do not expect the conclusion of
the  theorem to hold in the generality of
\cite{GKh2}, additional special cases such as operators with almost-periodic potentials
merit further consideration.
\end{rmk}

The theorem implies that the local eigenvalue statistics of $H_N(g)$ in the
regime $g < \gamma(\lambda_j(0))$ are the same as for the Hermitian operator $H_N(0)$.
\begin{cor}\label{cor}
In the setting of Theorem~\ref{thm:realev}, assume that for $g=0$
\begin{equation}\label{eq:bpois} \sum_{j=1}^N \delta_{(\lambda_j(g) - E) N \rho} \,\,\overset{\operatorname{distr}}\longrightarrow
\,\,\text{standard Poisson process} \end{equation}
for some $E \in \mathbb R$, $\rho > 0$, as $N \to \infty$. Then (\ref{eq:bpois}) holds for all  $0 \leq g < \gamma(E)$.
\end{cor}

In the Hermitian case $g=0$, a limit theorem of the form (\ref{eq:bpois}) was first proved by Molchanov \cite{Mol}
for a class of (continual) one-dimensional  Hermitian random Schr\"odinger operators.
An extension to higher-dimensional operators in the regime of Anderson localisation
was proved by Minami \cite{Min}; his result implies that (\ref{eq:bpois})
holds (for $g=0$) if the cumulative distribution function of $v_1$ is uniformly  Lipschitz, with
\begin{equation}\label{eq:ids}
 \rho = \rho(E) = \mathcal N'(E)~, \quad  \mathcal N(E) = \lim_{N \to \infty}
 \frac1N \# \{ \lambda_j(0) < E\}~.
\end{equation}
The existence of the density of states $\rho(E) = \mathcal N'(E)$ 
(in the sense of Radon) in this situation follows from an argument of Wegner \cite{Weg}.

Recently, Bourgain showed \cite{B3} for the one-dimensional case
that the density of states exists (and in
fact $\mathcal N$ is $C^\infty$ smooth) whenever the cumulative distribution function of $v_1$ is
uniformly H\"older continuous of some order $\nu > 0$. In \cite{B}, he showed
that (\ref{eq:bpois})   holds (for $g=0$) under the same assumptions, for the case of Dirichlet boundary
conditions (i.e.\ the top-right and bottom-left corner matrix elements are set to zero). The argument of \cite{B} can be adjusted to periodic boundary conditions (i.e.\ to $H_N(0)$). Combining this  with Corollary~\ref{cor}, we obtain that,  under the same assumption,   (\ref{eq:bpois}) also holds for all $g < \gamma(E)$ (at least, if $v_1$ is bounded
almost surely).

\paragraph{The logical structure of the paper}
 The key ingredient in the proof of Theorem~\ref{thm:realev} is a uniform lower bound
on the spectral radius of the transfer matrices associated with the Hermitian matrices $H_N(0)$, outside exponentially small neighbourhoods of the bands. This bound, possibly of independent interest, is stated as Proposition~\ref{prop:main} in Section~\ref{s:key}, where we also provide its proof. In Section~\ref{s:pf} we use it to prove Theorem~\ref{thm:realev}. 

The proof of Proposition~\ref{prop:main} makes use of several facts from the theory of random matrix products: particularly, a large deviation bound for the norm (Lemma~\ref{l:largedev}) and a comparison between the norm and the spectral radius (Lemma~\ref{l:radnorm}). While such statements are well known (the former goes back to the work of Le Page \cite{LePage2}, and the latter -- to the work of Guivarc'h \cite{G} and Reddy \cite{R}), the form in which we found them (and particularly the latter one) in the  literature is somewhat weaker than what is needed for our purposes. Therefore we develop, in Sections~\ref{s:rho} and~\ref{s:add}, an approach (close in spirit to the article \cite{SVW} of Shubin--Vakilian--Wolff and to unpublished work of the first author on the central limit theorem for eigenvalues of random matrix products) which allows us to re-prove these statements in the required form. 
Two other important ingredients of the proof of Proposition~\ref{prop:main} are Lemmata~\ref{l:bourgain} and \ref{l:lepage}, due to Bourgain \cite{B} and Le Page \cite{LePage}, respectively. The latter lemma lies in the field of random matrix products, and in the short Section~\ref{s:lepage}, we deduce it from Lemma~\ref{l:largedev}.

In  Section~\ref{s:prel}, we formulate the definitions and the lemmata required to state and prove Proposition~\ref{prop:main}. Possible generalisations and extensions of Theorem~\ref{thm:realev} are discussed in Section~\ref{s:comments}.

\section{Preliminaries}\label{s:prel}

\paragraph{Transfer matrices}
Let $E \in \mathbb R$. For $N =1,2,\cdots$, define
\[ \Phi_N(E) = T_N(E) \cdots T_2 (E)T_1(E)~,\]
where
\begin{equation}\label{eq:transf} T_j(E)  = \left(\begin{array}{cc} E - v_j& -1 \\ 1 & 0\end{array}\right) \in SL_2(\mathbb R)~. \end{equation}
More generally, one may consider the matrices $T_j(z)$ and $\Phi_N(z)$
for $z \in \mathbb C$. As usual, $\Phi_N(z)$ is associated to the formal 
solutions $\psi$ of the equation 
\[ \psi_{j-1} + v_j \psi_j + \psi_{j+1} = z \psi_j~, \quad j \geq 1\]
as follows:
\[ \Phi_N \begin{pmatrix} \psi_1 \\ \psi_0 \end{pmatrix} = \begin{pmatrix} \psi_{N+1} \\ \psi_N \end{pmatrix}~. \]

Denote 
\begin{equation}\label{eq:gamma1}
\gamma(E) = \lim_{N \to \infty} \frac1N \mathbb{E} \log \|\Phi_N(E)\|~.
\end{equation}
According to a result of Furstenberg and Kesten \cite{FK}, for any stationary ergodic sequence $v=(v_j)$, the following equality holds with probability one for any
(fixed) $E$:
\begin{equation}\label{eq:gamma2}
\lim_{N \to \infty} \frac1N \log \|\Phi_N(E)\| = \gamma(E)~.
\end{equation}
We emphasise that (\ref{eq:gamma2}) does not hold simultaneously for all $E$ (see \cite{G1}); in fact, in the i.i.d.\ case the left-hand side of (\ref{eq:gamma1}) vanishes
on a  dense random subset of $\mathcal S$.

\smallskip
A fundamental fact which is crucial for our considerations is the 
positivity of the Lyapunov exponent: in the i.i.d.\ case, Furstenberg's theorem \cite{F} implies that $\gamma(E) > 0$ for all $E \in \mathbb{R}$. (Formally, we use the quantitative version (\ref{eq:svw}) of this fact.)

\paragraph{Large deviations}
Large deviations bounds for the norm of a random matrix
product go back to the work of Le Page \cite{LePage2}. There are numerous extensions, see particularly the recent work \cite{Sert}, where a large deviation principle was obtained, and
references therein. We need the following upper bound, close to the original work of Le Page;
a proof is provided in Section~\ref{s:largedev}.
\begin{lemma}[Le Page]\label{l:largedev}
If $\mathbb E |v_1|^\eta< \infty$ for some $\eta > 0$, then for any $R > 0$ there exist $C,c>0$ such that for $\epsilon \in (0, 1/e)$ and $|E| \leq R$
\begin{equation}\label{eq:largedev} \mathbb{P} \left\{ \left| \frac1N \log \|\Phi_N(E)\| - \gamma (E) \right| \geq \epsilon \right\} \leq \left(C \log \frac1\epsilon\right) \exp\left[- \frac {c\epsilon^2N}{\log \frac1\epsilon}\right]~.\end{equation}
\end{lemma}

\paragraph{Spectral radius}
We shall use the following lemma, which is a variant of the results proved by Guivarc'h
\cite[\S 2.4]{G} and by Reddy \cite{R}. We provide a proof in Section~\ref{s:radnormpf}.
For a matrix $\Phi$ we denote by $\rho(\Phi)$ its spectral radius.
\begin{lemma}\label{l:radnorm}
If $\mathbb E|v_1|^\eta < \infty$ for some $\eta > 0$, then for any $R > 0$
there exist $B,B'>0$ and $b,b'>0$ such that
\[ \sup_{|E| \leq R} \mathbb{P} \left\{ \rho(\Phi_N(E)) \leq \delta \|\Phi_N(E)\| \right\}
\leq B \delta^{b} + B' e^{-b' N}~, \quad 0 \leq \delta \leq 1~.\]
\end{lemma}

\paragraph{Bands and gaps}
Consider the transfer matrices corresponding to the potential $(v_j)$ 
 (in this paragraph the potential does not have to be random).
The set of $E$ such that $\rho(\Phi_N(E)) = 1$ consists of $N$ disjoint intervals
(bands); we denote their interiors, numbered from the rightmost to the leftmost, 
$I_1, \cdots, I_N$. Denote
\[ \mathbb R = G_0 \uplus I_1 \uplus G_1 \uplus I_2 \uplus \cdots \uplus I_N \uplus G_N~, \]
where the $G_j$ (the closures of the gaps) are also ordered from right to left.

The eigenvalues of the periodic operator $H_N(0)$ are exactly the points $E$
at which $1$ is an eigenvalue of $\Phi_N(E)$. These are exactly
the edges of the gaps $G_0, G_2, \cdots$
with even indices. This fact admits the following generalisation to the non-Hermitian case (cf.\  \cite{GKh1,GKh2}):

\begin{lemma}[{\cite[Lemma 4.1]{GKh2}}]\label{l:eq}
The eigenvalues of $H_N(g)$ are the points $z \in \mathbb{C}$ such that $e^{Ng}$ is
an eigenvalue of $\Phi_N(z)$. \end{lemma}

\paragraph{H\"older continuity of the Lyapunov exponent}

The local H\"older continuity of the Lyapunov exponent goes back to
the work of Le Page \cite{LePage}. We need the
following version, proved in \cite{CKM} and, by different arguments, in \cite{SVW,B2};
for the sake of unity of argument, we provide a proof in Section~\ref{s:lepage}.
\begin{lemma}[Le Page]\label{l:lepage} If $v_j$ are independent,
identically distributed with $\mathbb{E} |v_1|^\eta < \infty$ for some $\eta > 0$, then the
Lyapunov exponent $\gamma(E)$ associated to the sequence $T_N(E)$ is
 uniformly H\"older continuous on any compact interval.
\end{lemma}
\bigskip
\paragraph{Gaps between the eigenvalues}

\begin{lemma}[Bourgain \cite{B}]\label{l:bourgain}If $v_j$ are independent, identically distributed with $\mathbb{E} |v_1|^\eta < \infty$ for some $\eta > 0$, then for any $R > 0$
there exists $K > 0$ such that
\[ \mathbb{P} \left\{ \min_{\text{$\lambda\neq\lambda'$--eigenvalues of $H_N$ in $[-R, R]$}} |\lambda - \lambda'| < N^{- K}\right\} = 0~.\]
\end{lemma}

\begin{rmk}In the work \cite{B} the lemma is proved for Dirichlet rather than periodic boundary conditions, and only for the case of Bernoulli potential. However, the argument presented there applies equally well in the current setting.
\end{rmk}

\begin{rmk}
The argument in \cite{B} relies on Anderson localisation. On the other hand, if the cumulative distribution function of $v_1$ is uniformly H\"older of order $\nu > 1/2$,
the conclusion of the lemma also follows from the Minami estimate, see \cite{Min}
and further \cite{CGK,Graf}. Thus, for such potentials, the conclusion of Theorem~\ref{thm:realev} is established using fixed-energy arguments only. \end{rmk}

\section{Proof of the main result}

\subsection{The key technical statement}\label{s:key}
Let $c>0$ be a sufficiently small constant, to be chosen later. For a gap $G_j = [a_j, b_j]$, denote
\[ G_j^{+,c} = [b_j - e^{-cN}, b_j]~, \quad  G_j^{-,c} = [a_j, a_j+ e^{-cN}]~. \]
The following proposition provides 
uniform control of the transfer matrices outside exponentially small neighbourhoods of the bands. It is the key ingredient in the proof of the main theorem. Having in mind possible additional applications (in the Hermitian and non-Hermitian setting), we formulate it as an independent statement.
\begin{prop}\label{prop:main}
Let $v_j$ be  i.i.d.\ with $|v_1| \leq A$ for some $A > 0$. Then
for any $\epsilon > 0$
\begin{equation}\label{eq:mainp1}\lim_{N \to \infty} \mathbb{P} \left\{ \forall 1 \leq j < N\,\,\, \max_{E \in G_j} \frac{1}{N} \log \rho(\Phi_N(E)) \geq
\max_{E \in G_j} \gamma(E) - \epsilon \right\} = 1~. \end{equation}
In addition, if $c > 0$ is small enough and $\mathfrak s \in \{+,-\}$,
\begin{equation}\label{eq:mainp2}\lim_{N \to \infty} \mathbb{P} \left\{ \forall 1 \leq j < N\,\,\, \max_{E \in G_j^{\mathfrak s, c}} \frac{1}{N} \log \rho(\Phi_N(E)) \geq
\max_{E \in G_j} \gamma(E) - \epsilon \right\} = 1~. \end{equation}
\end{prop}

\begin{proof}
As customary, we denote  $\lfloor x \rfloor = \max_{n \leq x} n$ and $\lceil y \rceil
= \min_{n \geq x} n$. 

If $|v_1| \leq A < \infty$ almost surely, then $\| H_N(0) \| \leq 2 + A$.
Let $(E_l)_{l=1}^{\lfloor e^{cN} \rfloor}$ be a sequence of equally spaced points with
$E_1 = -3-A$ and $E_{\lfloor e^{cN} \rfloor} = 3+A$.

The large deviation estimate of Lemma~\ref{l:largedev} allows to bound the norm of the transfer matrix from below, outside an event of small probability. Formally, for any $E_l\in [-3-A,3+A]$, we have
\[
\mathbb P \left\{ \, \|\Phi_N(E_l)\| \leq e^{N(\gamma(E_l) - \frac{\epsilon}4)} \right\}
\le B_1 e^{-b_1 N},
\]
where $B_1$ and $b_1$ do not depend on $E_l\in [-3-A,3+A]$.

In turn,  Lemma~\ref{l:radnorm} allows to compare the spectral radius with the norm: taking $\delta=e^{-N\frac{\epsilon}{4}}$ in the lemma, we obtain for any $E_l$:
\[
\mathbb P \left\{ \, \rho(\Phi_N(E_{l})) \leq e^{-N\frac{\epsilon}{4}} \|\Phi_N(E_l)\| \right\}
\le Be^{-Nb\frac{\epsilon}{4}}+B'e^{-b' N}\le Ce^{-\bar{b} N},
\]
where $C>0$ and $\bar{b}>0$ do not depend on $E_l\in [-3-A,3+A]$.
Hence
\begin{equation*}\begin{split}
&\mathbb P \left\{ \exists l: \, \,  \frac{1}{N} \log \rho(\Phi_N(E_l))  \leq  \gamma(E_l) - \frac{\epsilon}2 \right\} \\
&\quad\leq e^{cN} \left[ B_1 e^{-b_1 N}+ C e^{-\bar{b}N}  \right] \leq C' e^{-c'N}~,
\end{split}\end{equation*}
where we chose $c>0$ small enough ($c < \min(b_1, \bar{b})$). In particular, the probability of the event
\[ \Omega_1 = \left\{ \forall l \, \, \frac{1}{N} \log \rho(\Phi_N(E_l))  \geq  \gamma(E_l) - \frac{\epsilon}2  \right\} \]
tends to $1$ as $N \to \infty$.

Let $m_j \in G_j$ be such that $\max\limits_{E \in G_j} \gamma(E) = \gamma(m_j)$.
Denote
\[ \Omega_2 = \left\{ \forall 1 \leq j < N \,\exists l_j: \, E_{l_j} \in G_j,
|E_{l_j} - m_j| \leq C e^{-cN} \right\}~.\]
By Lemma~\ref{l:bourgain}, no gap is exponentially short, hence each $G_j$
contains at least one $E_l$. Therefore
also the probability of $\Omega_2$ tends to $1$, for $C>0$ large enough.
Then on the event
$\Omega_1 \cap \Omega_2$, for
sufficiently large $N$:
\[\begin{split}
\forall j \quad \max_{E \in G_j} \frac1N \log \rho(\Phi_N(E)) &\geq
\frac1N \log \rho(\Phi_N(E_{l_j}))  \\
&\geq \gamma(E_{l_j})-\frac{\epsilon}2 \geq \gamma(m_j)-\epsilon
= \max_{E \in G_j} \gamma(E) - \epsilon~,
\end{split}\]
where we used the H\"older continuity of the Lyapunov exponent (Lemma~\ref{l:lepage}).
Similarly,
\[ \forall j \quad \max_{E \in G_j^{\mathfrak s, c}} \frac1N \log \rho(\Phi_N(E)) \geq  \max_{E \in G_j} \gamma(E) - \epsilon~.\]
\end{proof}

\subsection{Proof of Theorem~\ref{thm:realev}}\label{s:pf}

Let $\epsilon > 0$. By  equation (\ref{eq:mainp1}) of Proposition~\ref{prop:main},
\[\lim_{N \to \infty} \mathbb P \Big\{ \forall 1 \leq j < N\,\,\, \max_{E \in G_j} \frac1N \log \rho(\Phi_N(E)) \geq
\max_{E \in G_j} \gamma(E) - \epsilon/2 \Big\} = 1~. \]
On this event, we have for any $1 \leq j < N$:
\[ \max_{E \in G_j} \frac1N \log \rho(\Phi_N(E)) \geq\max_{E \in G_j} \gamma(E) - \epsilon/2~.\]
Since $G_{2j} = [\lambda_{2j+1}, \lambda_{2j}]$ (where $\lambda_j = \lambda_j(0)$;
in the notation of Proposition~\ref{prop:main}, $a_{2j} = \lambda_{2j+1}$
and $b_{2j} = \lambda_{2j}$), we conclude the following:
if, for some $g$,
\[\max(\gamma(\lambda_{2j}), \gamma(\lambda_{2j+1})) \geq g + \epsilon~, \]
then
\[ \max_{E \in G_{2j}}\frac1N \log \rho(\Phi_N(E)) \geq  (g+\epsilon/2)~. \]
 On the other hand,
$\frac1N \log \rho(\Phi_N(\lambda_{2j})) = \frac1N \log \rho(\Phi_N(\lambda_{2j+1}))
= 0$, hence by the intermediate
value theorem there are two solutions to $\rho(\Phi_N(E)) = e^{Ng}$ lying in $G_j$.
By Lemma~\ref{l:eq}, these are exactly the eigenvalues $\lambda_{2j}(g)$ and $\lambda_{2j+1}(g)$.
As to the eigenvalues $\lambda_1(g)$ and, for even $N$, $\lambda_N(g)$, these
are real. Invoking the second part (\ref{eq:mainp2}) of Proposition~\ref{prop:main}, we obtain that $|\lambda_j(g) - \lambda_j(0)|$ is exponentially small.
\qed

\section{On the spectral radius of transfer matrices}\label{s:rho}

The ultimate goal of this section is the proof of Lemma~\ref{l:radnorm} in Section~\ref{s:radnormpf}. We start with some auxiliary statements.

Let $T_j  = T_j(E)$ with $\mathbb E |v_j|^\eta< \infty$ for some $\eta > 0$. Then
$\mathbb{E} \|T_j\|^\eta < e^{A\eta}$, where $A>0$ can be chosen locally
uniformly in $E$.
Let $\Phi_N = T_N \cdots T_2 T_1$, and
further let $\Phi_{N, M} = T_N \cdots T_{M+1}$ for $N > M$. We use the singular value
decomposition
\begin{equation}\label{eq:usv} \Phi_N = U_N \left(\begin{array}{cc} s_N & 0 \\ 0 & s_N^{-1} \end{array} \right)V_N~, \quad \Phi_{N,M} = U_{N,M} \left(\begin{array}{cc} s_{N,M} & 0 \\ 0 & s_{N,M}^{-1} \end{array} \right)V_{N,M}~, \end{equation}
where $s_N, s_{N,M} \geq 1$, $U_N,V_N,U_{N,M},V_{N,M} \in SO(2)$. The application
of singular value decomposition in the study of random matrix products goes back at least to
the work of Tutubalin \cite{Tut}, who realised that the sequence $(U_N)$ is approximated by a Markov chain whereas $V_N$ converges to a random limit. This idea plays an important r\^ole in our analysis as well.

\subsection{A lemma in linear algebra}
Denote by $uv^*$ the rank-one operator
taking $w$ to $(w, v)u$, where $(u, v)$ is the inner product. Also, we denote by $e_j$ the $j$-th vector of the standard basis. Although we need the following lemma only for two-dimensional matrices, specialising the argument to this case would only obscure the idea.
\begin{lemma}\label{l:linalg}
If $u, v \in \mathbb C^m$ and $\mathfrak h: \mathbb C^m \to \mathbb C^m$ is a
linear map such that $\|\mathfrak h\| \leq \frac{|(u,v)|^2}{9\|u\|\|v\|}$, then
\begin{equation}\label{eq:linalg}
\rho(uv^* + \mathfrak h) \geq \frac12|(u,v)|~.\end{equation}
\end{lemma}

\begin{proof}
Let $0 \leq t \leq 1$, and let $z$ be a complex number on the circle of
radius $\frac12 |(u,v)|$ about $(u,v)$.
We shall show that for such $t$ and $z$ the determinant
$\det (z - uv^* - t\mathfrak h)$
does not vanish. This will imply that the number of eigenvalues of
$uv^* + t\mathfrak h$ in the disc enclosed by the circle does not change as
$t$ varies from $0$ to $1$. For $t = 0$, the spectrum of $uv^*$
consists of two eigenvalues, $0$ (with multiplicity $m-1$) and $(u,v)$ (with multiplicity $1$), of which the second one
lies in the disk; thus also for $t=1$ there is (exactly) one simple eigenvalue in
the disc, and in particular (\ref{eq:linalg}) holds.

Let us factorise
\[ \det (z\mathbbm{1} - uv^* - t\mathfrak h)
= \det (z\mathbbm{1} -uv^*) \, \det( \mathbbm{1} - (z\mathbbm{1} - uv^*)^{-1} t\mathfrak h)~.\]
The first term is equal to
\[ z^{m-1} (z -(u,v))~, \]
and thus does not vanish on the circle. To show that the second term
does not vanish,  observe that
\[\begin{split} \| (z\mathbbm{1} - uv^*)^{-1}\|
&= \|\frac1z \mathbbm{1} + \frac1z \frac{uv^*}{z - (u,v)}\|
\leq \frac{1}{|z|} \left\{ 1 + \frac{\|u\| \|v\|}{|z -(u,v)|} \right\} \\
&\leq \frac{2}{|(u,v)|} \frac{4\|u\|\|v\|}{|(u,v)|}
= \frac{8\|u\|\|v\|}{|(u,v)|^2}~, \end{split}\]
hence
\[ \| (z \mathbbm{1}- uv^*)^{-1} t\mathfrak h \| \leq \frac{8\|u\|\|v\|}{|(u,v)|^2}  \frac{ |(u,v)|^2}{9\|u\|\|v\|} = \frac89 < 1~.\]
\end{proof}

\begin{cor}\label{cor:linalg} If $\Phi = U \left( \begin{array}{cc} s & 0 \\ 0 & 1/s \end{array}\right) V$
with $U,V \in SO(2)$ and $s \geq 3$, $|(VUe_1, e_1)| \geq \frac 3 s$, then
\[ \rho(\Phi) \geq \frac s 2 |(VUe_1, e_1)|~. \]
\end{cor}

\begin{proof}
Apply the lemma to $u = Ue_1$, $v = V^* e_1$, observing that
\[ U \left( \begin{array}{cc} s & 0 \\ 0 & 0 \end{array}\right) V = u v^*\]
and 
\[ \| U \left( \begin{array}{cc} 0 & 0 \\ 0 & 1/s \end{array}\right) V \| = 1/s~. \qedhere\]
\end{proof}

Our goal in the remaining part of the section is to prove the estimate
\begin{equation}\label{eq:red1}
 \mathbb P \left\{ |(V_N U_N e_1,e_1)| \leq \delta  \right\}
\leq K\delta^\kappa + e^{-cN}~,
\end{equation}
which will imply Lemma~\ref{l:radnorm}, in view of Corollary~\ref{cor:linalg}.

\subsection{On an important unitary operator}
For each $T \in SL_2(\mathbb R)$, consider the operator $\pi(T): L_2(S^1) \to
L_2(S^1)$, defined via
\[ (\pi(T)f)(x) = f(Tx/\|Tx\|)/\|Tx\|~. \]
For any $T \in SL_2(\mathbb R)$, $\pi(T)$ is unitary.
\begin{lemma}[Shubin--Vakilian--Wolff \cite{SVW}]\label{l:norm} If $v_1$ is not
almost surely equal to a constant, then there exists $a>0$ such that
\begin{equation}\sup_{E} \| (\mathbb E \pi(T_1(E)))^2 \|\leq e^{-a} < 1~.\end{equation}
\end{lemma}
\noindent Denoting by $\mathbf 1 \in L_2(S^1)$ the function identically equal to $1$ and 
parametrising the points on the circle by an argument $\theta \in [0, 2\pi]$,
we obtain:
\begin{equation}\label{eq:svw}\mathbb{E} \|\Phi_n(E)\|^{-1} \leq \mathbb{E} \int_0^{2\pi} \frac{d\theta}{2\pi} \| \Phi_n(E) e^{i \theta}\|^{-1} = \frac{1}{2\pi}
(\mathbb E \pi(\Phi_n) \mathbf 1, \mathbf 1)  \leq e^{-a\lfloor \frac n 2 \rfloor}~, \end{equation}
and in particular $\gamma \geq a/2$.

\medskip
Next, by the Oseledec multiplicative ergodic theorem \cite{Os}, $V_N$ converges almost
surely to a random limit $V$. We use this fact in the following
form:
\begin{lemma}\label{l:os} Suppose $\mathbb{E} |v_1|^\eta < \infty$ for some $\eta > 0$.
Then, for any $R > 0$,
\[ \sup_{|E| \leq R} \mathbb{E} \|V_N(E) - V(E)\| \leq B e^{-bN}~. \]
\end{lemma}

\begin{proof}
Let $V_N^* e_2 = \cos \theta \, V_{N-1}^* e_2 +  \sin \theta \, V_{N-1}^* e_1$. Then
\[\begin{split}
\|\Phi_N\|^{-1}  &= \| \Phi_{N} V_N^* e_2\| = \| T_N \Phi_{N-1} V_N^* e_2\| \geq
\|T_N\|^{-1} \| \Phi_{N-1} V_N^* e_2 \| \\
&\geq \|T_N\|^{-1} \left[ \cos^2 \theta \| \Phi_{N-1}\|^{-2} + \sin^2 \theta \|\Phi_{N-1}\|^2\right]^{1/2} \\
&\geq |\sin \theta| \, \|T_N\|^{-1}  \|\Phi_{N-1}\|~,\end{split}\]
whence
\[ |\sin \theta| \leq \frac{\|T_N\|}{\|\Phi_N\| \, \|\Phi_{N-1}\|}~, \]
and, using (\ref{eq:svw}), 
\[ \mathbb E  |\sin \theta|  \leq \mathbb E |\sin \theta|^\eta \leq e^{A(R)\eta} \, B_0 e^{-bN}~, \]
Thus an estimate of the same form holds for $\|V_N - V_{N-1}\|$ in place of
$|\sin \theta|$.
\end{proof}

\subsection{Conclusion of the proof of Lemma~\ref{l:radnorm}}\label{s:radnormpf}

As before, we denote  $\lfloor x \rfloor = \max_{n \leq x} n$ and $\lceil y \rceil
= \min_{n \geq x} n$. 
In the notation of (\ref{eq:usv}), Lemma~\ref{l:os} applied to the matrix products
$T_N \cdots T_1$ and $T_1^* T_2^* \cdots T_N^*$ implies that for $N \geq C \log A$
\begin{equation}\label{eq:os}\mathbb{E} \| V_N - V_{\lfloor N/2\rfloor}\| \leq  e^{-b'N}~,
\quad
\mathbb{E} \| U_N - U_{N,\lceil N/2\rceil}\| \leq e^{-b'N}~.
\end{equation}
The matrices $V_{\lfloor N/2\rfloor}, U_{N,\lceil N/2\rceil}$ are independent, therefore
by an additional application of Lemma~\ref{l:os}
\begin{equation}\label{eq:os'}\mathbb{E} \| V_N - \tilde V\| \leq   e^{-b'N}~,
\quad
\mathbb{E} \| U_N - \tilde U\| \leq   e^{-b'N}~,
\end{equation}
where $\tilde U, \tilde V \in SO(2)$ are independent random matrices sampled from the corresponding limiting distributions. To conclude the proof of Lemma~\ref{l:radnorm}, we state (and prove)

\begin{lemma}\label{l:inv}
Assume that $\mathbb E |v_1|^\eta < \infty$, and that $v_1$ is not almost surely constant.
Then for any $E$ there exist $K>0$ and $\kappa>0$ such that for any $0 \leq \delta \leq 1$
and $w, w' \in S^1$
\begin{equation}\label{eq:reg}
\mathbb{P} \left\{ |(U(E)w, w')| \leq \delta \right\} \leq K \, \delta^{\kappa}~, \quad \mathbb{P} \left\{ |(V(E)w, w')| \leq \delta \right\} \leq K \, \delta^{\kappa}~.
\end{equation}
The numbers $K$ and $\kappa$ may be chosen
locally uniformly in $E$.
\end{lemma}

\begin{proof}[Proof of Lemma~\ref{l:radnorm}]
Applying Lemma~\ref{l:inv} with $\tilde U$ and $\tilde V$ in place of $U$ and $V$, we obtain:
\begin{equation}\label{eq:invmes}
 \mathbb P \left\{ |(\tilde V \tilde Ue_1, e_1) | \leq \delta  \right\}
\leq K \, \delta^{\kappa}~.
\end{equation}
Together with (\ref{eq:os'}), this implies (\ref{eq:red1}), which by Corollary~\ref{cor:linalg}
implies the conclusion of Lemma~\ref{l:radnorm}. \end{proof}

\begin{proof}[Proof of Lemma~\ref{l:inv}]
Since we do not keep track on the dependence of the constants on $a$ (from
Lemma~\ref{l:norm}) and $\eta$, we
may assume that $a < 4\eta$. Let $\mathcal C_\delta \subset S^1$ be a cap of angular size $\delta$. Assume: $\mathbb{E} \|T_1\|^\eta \leq e^{A\eta}$, and denote $n = \lfloor \log \frac {2(A+1)}{\delta} \rfloor$. Let
us show that
\begin{equation}\label{eq:goalinv}
\mathbb P \left\{ V_n^* e_2  \in \mathcal C_\delta \right\} \leq C \lceil \log \frac{16(A+1)}a \rceil e^{-\frac{an}{16}}~.
\end{equation}
By an additional application of Lemma~\ref{l:os}, this implies (\ref{eq:reg}).

\noindent To prove (\ref{eq:goalinv}), we start with the estimates
\[\mathbb{E} \|\Phi_n\|^\eta \leq e^{An\eta}~, \quad
\mathbb{E} \|\Phi_n\|^{-1} \leq (\mathbb E \pi(\Phi_n) \mathbf 1, \mathbf 1)
\leq 2\pi e^{-\frac {an}{2}}~,\]
which imply that
\[ \mathbb{P} \left\{ e^{\frac{an}{4}} \leq \|\Phi_n\|\leq e^{(A+1)n} \right\} \geq 1 - C_1 e^{-\frac{an}{4}}~.\]
If (\ref{eq:goalinv}) fails, there exists $4 \leq m \leq \lceil \log \frac{16(A+1)}{a} \rceil$
such that
\begin{equation}\label{eq:linv5}\mathbb P \left\{ e^{\frac{ma}{16} n} \leq \| \Phi_n \| \leq e^{\frac{(m+1)a}{16}n}~, \,\,\,
V_n^* e_2 \in \mathcal C_\delta \right\} \geq C e^{-\frac{an}{16}}~.\end{equation}
Let $\tau = \max(\delta, e^{- \frac{(m+1)a}{16}n})$, let $\mathcal C_{\tau} \supset \mathcal C_\delta$ be a cap of size $\tau$, and let $\mathbf 1_{\tau}$ be the indicator of $\mathcal C_{\tau}$. Then (still treating $e^{i\theta}$ as a vector on the circle)
\begin{equation}\label{eq:linv4}
 \mathbb E \int_{\mathcal C_{\tau}} \|\Phi_n e^{i\theta}\|^{-1} d\theta
 =(\mathbb E \pi(\Phi_n) \mathbf 1, \mathbf 1)
 \leq C_2 e^{-\frac{an}{2}} \sqrt{\tau}~.
\end{equation}
On the other hand, if
\begin{equation}\label{eq:linv2}
V_n^* e_2 \in \mathcal C_\tau~, \,\,\, e^{\frac{ma}{16} n} \leq \| \Phi_n \| \leq e^{\frac{(m+1)a}{16}n}~,
\end{equation}
then
\begin{equation}\label{eq:linv1}
\int_{\mathcal C_{\tau}} \|\Phi_n e^{i\theta}\|^{-1} d\theta
\geq \frac{1}{C_3} \min(\tau, \|\Phi_n\|^{-2}) \|\Phi_n\| \geq \frac{1}{C_3} \tau e^{\frac{(m-1)a}{16}n}~,
\end{equation}
therefore by (\ref{eq:linv5})
\begin{equation}\label{eq:linv3}\begin{split}
\mathbb \int_{\mathcal C_{\tau}} \|\Phi_n e^{i\theta}\|^{-1} d\theta
&\geq \frac{1}{C_3} \tau \|\Phi_n\| \geq \frac{1}{C_3} \tau e^{\frac{(m-1)a}{16}n}
\,\,  \mathbb P(\text{the event (\ref{eq:linv2})}) \\
& \geq \frac{1}{C_3} \tau e^{\frac{(m-1)a}{16}n}  Ce^{-\frac{an}{16}}~.
\end{split}\end{equation}
If we choose $C > C_2 C_3$, the juxtaposition of   (\ref{eq:linv3}) with (\ref{eq:linv4})
leads to
\[ e^{- \frac{(m+1) a}{16}n}\leq \sqrt \tau \leq e^{-\frac{(m-1)a}{16}n} e^{\frac{an}{16}} e^{-\frac{an}{2}}~,  \]
which is a contradiction.
\end{proof}

\section{Proofs of the additional lemmata}\label{s:add}

\subsection{Large deviations: proof of Lemma~\ref{l:largedev}}\label{s:largedev}

We suppress the dependence on the spectral parameter $E$, on which the estimates
below are locally uniform. Fix $x \in S^1$. It will suffice to prove the following: for any $\epsilon \in (0, 1/e)$
\[
\mathbb{P} \left\{ \left| \frac{1}{N} \log \| \Phi_N x\| - \gamma\right| \geq \epsilon \right\}
\leq C \log \frac1\epsilon \exp\left\{-\frac{c\epsilon^2 N}{\log \frac1\epsilon}\right\}~. \]
Let $x_0 = x$, and let $x_{j+1} = T_{j+1} x_j / \|T_{j+1} x_j\|$. The vectors $x_j \in S^1$ form a Markov chain, and
\[ \| \Phi_N x\| = \prod_{j=1}^N \|T_j x_{j-1}\|~. \]
Our strategy from this point (based on two arguments going back to the
work of S.\ N.\ Bernstein \cite{Bern1,Bern}) is as follows. Fix $k = \lceil \log \frac1\epsilon\rceil$, 
and split the product into $k$ sub-products corresponding to the different
residues of $j$ modulo $k$. The terms in each sub-product are almost independent;
we make them independent by restarting the Markov chain from an invariant
distribution every $k$ steps. Then we obtain a bound on the positive and
negative fractional moments of each sub-product, from which the desired estimate
follows using the Chebyshev inequality.

Formally, for each $j$ choose (independently) a random vector
on the circle, distributed according to the invariant measure of the Markov
chain; denote this vector by $y_{j-k,j}$. Then  set 
$y_{n+1,j} = T_{n+1} y_{n,j}/ \|T_{n+1} y_{n,j}\|$ for  $n \geq j-k$,
and, finally, define $y_j = y_{j, j}$. Then, for each $r$, the random variables $\left\{ y_j \, \mid \, j \in I_r\right\}$,
where
\[ I_r = (r + k \mathbb Z)
\cap \{1,\cdots,N\}~,\]
are jointly independent. The vectors $y_j$ are close to $x_j$:  for $j \geq k$,
\begin{equation}\label{eq:condos} \mathbb{E} \left\{ \|x_j - y_j  \| \, \mid \,
 v_1, v_2, \cdots,v_{j-k} \right\} \leq C \exp(-c k)~,
\end{equation}
as implied by  the following consequence of Lemma~\ref{l:inv}:
\[ \mathbb{P} \left\{ V_{j,j-k}^*e_2 \in \mathcal C_\delta \right\} \leq C \left[ \delta^c + e^{-ck} \right]~. \]
Denote
\[ A_r = \prod_{j \in I_r} \|T_j x_{j-1}\|~, \quad \text{so that} \quad  \| \Phi_N x\| = \prod_{r=1}^k A_r~,\]
and observe that (for $0 \leq p \leq 1$)
\begin{equation}\label{eq:u1}
\|T_j x_{j-1}\|^{p} \leq  \|T_j y_{j-1}\|^{p} + \|T_j\|^p \|x_{j-1} - y_{j-1}\|^p\end{equation}
whereas
\begin{equation}\label{eq:l1}\begin{split} \|T_j x_{j-1}\|^{-p} &=
\|T_j y_{j-1}\|^{-p} + ( \|T_j x_{j-1}\|^{-p}  - \|T_j y_{j-1}\|^{-p} )  \\
&\leq \|T_j y_{j-1}\|^{-p} + \frac{ \|T_j y_{j-1}\|^{p}  - \|T_j x_{j-1}\|^{p}}{ \|T_j y_{j-1}\|^{p}  \|T_j x_{j-1}\|^{p}} \\
&\leq  \|T_j y_{j-1}\|^{-p} +  \|T_j\|^{3p} \|x_{j-1}- y_{j-1}\|^{p}~.
\end{split}\end{equation}
Also observe that (for $\eta$ from the formulation of the lemma)
\begin{equation}\label{eq:tmpmoments} \mathbb E \log \|T_j y_{j-1}\| = \gamma~, \quad
\mathbb E \|T_j y_{j-1}\|^{\pm\eta}
\leq C < \infty~.\end{equation}
From (\ref{eq:condos}) and (\ref{eq:tmpmoments}) we obtain that for $0 < p < \min(\frac{\eta}{10}, 1)$
\begin{equation}
\begin{split}
\mathbb{E} \left\{ \|T_{j} x_{j-1}\|^p \, \mid \, v_1, v_2, \cdots,v_{j-k} \right\}
&\leq 1 + p \gamma + C (p^2 + e^{-ck}) \\
\mathbb{E} \left\{ \|T_{j} x_{j-1} \|^{-p} \, \mid \, v_1, v_2, \cdots,v_{j-k} \right\}
&\leq 1 + p \gamma + C (p^2 + e^{-ck})~.
\end{split}
\end{equation}
Taking the products of each of these inequalities over $j \in I_r$
and using the exponential Chebyshev inequality, we have:
\begin{equation}
\mathbb P \left\{ \left| \log A_r - \gamma \# I_r \right| \geq \epsilon \right\}
\leq \exp(- c \epsilon^2 \# I_r)~,
\end{equation}
whence Lemma~\ref{l:largedev} follows by the union bound.
\qed

\begin{rmk} Using a slightly longer spectral-theoretic argument, one may
dispose of the logarithmic terms in (\ref{eq:largedev}).
\end{rmk}

\subsection{H\"older continuity: proof of Lemma~\ref{l:lepage} }\label{s:lepage}

Let $R \geq e$, $\delta \in (0, 1/e]$, and   $N = \lfloor \log \frac1\delta \rfloor$.
By the large deviation estimate (\ref{eq:largedev}), for any $|E|, |E'| \leq R$
\[ \mathbb{P} \left\{ \Big| \log \|\Phi_N(E)\| - \gamma(E)\Big| <  \frac{\delta}{3}~, \quad
\Big| \log \|\Phi_N(E')\| - \gamma(E')\Big| <  \frac{\delta}{3} \right\} \geq \frac34~.
\]
Next, by the assumption $\mathbb{E} |v_1|^\eta  < \infty$, we have for sufficiently
large $C$:
\[ \mathbb{P} \left\{ \| \Phi_j \| \geq (CR)^j \right\}  \leq e^{-10 j}~,\]
therefore with probability $> 3/4$
\[ \forall 1 \leq j \leq N \quad \| \Phi_j(E) \|, \|\Phi_j(E')\| \leq (CR)^j~, \quad
\| \Phi_{N,j} (E)\|, \|\Phi_{N,j} (E')\| \leq (CR)^{N-j}~,\]
and on this event
\[ \| \Phi_N(E) - \Phi_N(E') \| \leq |E-E'| N (CR)^N~.\]
Therefore for $|E-E'| \leq (2CR)^{-N}  = \delta^{\log(2CR)}$
\[ \left|\gamma(E) - \gamma(E')\right| \leq \delta~,\]
as claimed. \qed

\section{Outlook}\label{s:comments}

Let us briefly comment on possible extensions and directions for further study.

\paragraph{Other potentials}
As mentioned in Remark~\ref{rmk:ext}, it would be interesting to explore the counterparts
of Theorem~\ref{thm:realev} for other stationary (but non-independent) potentials.
It may be of independent interest to explore the counterparts of
Proposition~\ref{prop:main} in this setting.

\paragraph{Higher dimension} The arguments used in the proof of 
Theorem~\ref{thm:realev} can be recast into the language of resolvent estimates.
In this form, they are applicable to the following higher-dimensional analogue of (\ref{eq:defop}) acting on $\ell_2(\mathbb Z^d / N \mathbb Z^d)$:
\begin{equation}\label{eq:defop'}
\begin{split} 
(H_{N, d}(g) \psi) (x) &=  e^{-g} \psi(x+e_1) + e^g \psi(x-e_1) \\
&+ \sum_{j=2}^d (\psi(x+e_j) + \psi(x-e_j)) + v_x \psi(x)~, 
\end{split}\end{equation}
where $v_x$ are i.i.d. We state a sample result that can be proved by these
arguments.
\begin{prop}\label{thm:realev'} Assume that  the cumulative distribution function of $v_1$ is uniformly H\"older of order $\nu > 1/2$. Let $I \subset \mathbb R$ be a bounded interval such that, for some $\eta \in (0, 1)$,
\begin{equation}\label{eq:fr} \mathbb{E} |(H_{N,d}(0) - E)^{-1}(x, y)|^\eta < C \exp(-\gamma \eta \min(|x_1 - y_1|, N-|x_1 - y_1|))\end{equation}
for all $E$ in $I$. Then for any $g < \gamma$ 
there exists $c >0$ such that
\[
\mathbb{P} \Big\{ \forall j \,\,\, \text{s.t.} \,\,\,  \lambda_j(0) \in I:  \,\,\,\,
\lambda_j(g) \in \mathbb R \text{ and } |\lambda_j(g) - \lambda_j(0)| \leq e^{-c N} \Big\} \underset{N \to \infty}\longrightarrow 1~.\]
\end{prop}
The assumption (\ref{eq:fr}) is a signature of Anderson localisation; it was shown
by Aizenman and Molchanov \cite{AM} to hold for any interval $I$ when the disorder 
is sufficiently strong, and for intervals $I$ at the spectral edges for any strength of 
the disorder. A similar result can be proved if (\ref{eq:fr}) is replaced with the 
conclusion of the multiscale analysis of Fr\"ohlich and Spencer \cite{FS}. Proposition~\ref{thm:realev'} confirms the prediction of Kuwae and Taniguchi \cite{KT}, which
was challenged in some of the subsequent works (see \cite{Molinari} and references
therein).

Similarly to Theorem~\ref{thm:realev}, the proof of Proposition~\ref{thm:realev'} 
makes use of a mesh  $(E_l)_{l=1}^{\lfloor e^{cN} \rfloor}$ in $I$. Instead of 
Lemma~\ref{l:eq}, one relies on the following observation: if an interval
$(E_l, E_{l+1})$ between a pair of adjacent points of the mesh contains exactly one eigenvalue $\lambda_j(0)$ of $H_{N,d}(0)$, and,  for all $0 < g' < g$, the
points  $E_l, E_{l+1}$ are not eigenvalues of $H_{N,d}(g')$, then also 
$\lambda_j(g) \in  [E_l, E_{l+1}] \subset \mathbb R$.

\paragraph{Beyond the smallest Lyapunov exponent}
While in dimension $1$ the conclusion of Proposition~\ref{thm:realev'} is
similar to that of Theorem~\ref{thm:realev}, we emphasise a distinction between
resolvents and transfer matrices, which becomes essential already for a 
one-dimensional strip of width $\geq 2$: the decay of the resolvent kernel is only 
sensitive to the smallest Lyapunov exponent, whereas the full description of the 
eigenvalues of non-Hermitian operators of the form considered here is believed to 
depend on all the Lyapunov exponents. Similarly, in higher dimension, the matrices 
(\ref{eq:defop'}) are believed to have some real eigenvalues in the spectral
regions in which Anderson localisation does not hold.

\end{document}